\documentclass[10pt]{amsart}
\usepackage{amsmath,amssymb,graphicx}

\begin{document}

\newtheorem{thm}{Theorem}[section]
\newtheorem{lem}[thm]{Lemma}
\newtheorem{prop}[thm]{Proposition}
\newtheorem{cor}[thm]{Corollary}
\newtheorem{conj}[thm]{Conjecture}
\newtheorem{defn}[thm]{Definition}
\newtheorem*{remark}{Remark}

\numberwithin{equation}{section}

\newcommand{\Z}{{\mathbb Z}} 
\newcommand{\Q}{{\mathbb Q}}
\newcommand{\R}{{\mathbb R}}
\newcommand{\C}{{\mathbb C}}
\newcommand{\N}{{\mathbb N}}
\newcommand{\FF}{{\mathbb F}}
\newcommand{\T}{{\mathbb T}}
\newcommand{\fq}{\mathbb{F}_q}

\def\scrA{{\mathcal A}}
\def\scrB{{\mathcal B}}
\def\scrD{{\mathcal D}}
\def\scrE{{\mathcal E}}
\def\scrH{{\mathcal H}}
\def\scrK{{\mathcal K}}
\def\scrL{{\mathcal L}}
\def\scrM{{\mathcal M}}
\def\scrN{{\mathcal N}}
\def\scrS{{\mathcal S}}

\newcommand{\rmk}[1]{\footnote{{\bf Comment:} #1}}

\renewcommand{\mod}{\;\operatorname{mod}}
\newcommand{\ord}{\operatorname{ord}}
\newcommand{\TT}{\mathbb{T}}
\renewcommand{\i}{{\mathrm{i}}}
\renewcommand{\d}{{\mathrm{d}}}
\renewcommand{\^}{\widehat}
\newcommand{\HH}{\mathbb H}
\newcommand{\Vol}{\operatorname{vol}}
\newcommand{\area}{\operatorname{area}}
\newcommand{\tr}{\operatorname{tr}}
\newcommand{\norm}{\mathcal N} 
\newcommand{\intinf}{\int_{-\infty}^\infty}
\newcommand{\ave}[1]{\left\langle#1\right\rangle} 
\newcommand{\Var}{\operatorname{Var}}
\newcommand{\Cov}{\operatorname{Cov}}
\newcommand{\Prob}{\operatorname{Prob}}
\newcommand{\sym}{\operatorname{Sym}}
\newcommand{\disc}{\operatorname{disc}}
\newcommand{\CA}{{\mathcal C}_A}
\newcommand{\cond}{\operatorname{cond}} 
\newcommand{\lcm}{\operatorname{lcm}}
\newcommand{\Kl}{\operatorname{Kl}} 
\newcommand{\leg}[2]{\left( \frac{#1}{#2} \right)}  
\newcommand{\SL}{\operatorname{SL}}

\newcommand{\sumstar}{\sideset \and^{*} \to \sum}

\newcommand{\LL}{\mathcal L} 
\newcommand{\sumf}{\sum^\flat}
\newcommand{\Hgev}{\mathcal H_{2g+2,q}}
\newcommand{\USp}{\operatorname{USp}}
\newcommand{\conv}{*}
\newcommand{\dist} {\operatorname{dist}}
\newcommand{\CF}{c_0} 
\newcommand{\kerp}{\mathcal K}

\newcommand{\gp}{\operatorname{gp}}
\newcommand{\Area}{\operatorname{Area}}

\title{Quantum Chaos for point scatterers on flat tori}
\author{Henrik Uebersch\"ar}
\address{Institut de Physique Th\'eorique, CEA Saclay, 91191 Gif-sur-Yvette Cedex, France.}
\email{henrik.uberschar@cea.fr}
\date{\today}
\maketitle

\begin{abstract}
This survey article deals with a delta potential - also known as a point scatterer - on flat 2D and 3D tori. We introduce the main conjectures regarding the spectral and wave function statistics of this model in the so-called weak and strong coupling regimes. We report on recent progress as well as a number of open problems in this field.
\end{abstract}

\section{Introduction}

The point scatterer on a torus is a popular model to study the transition between integrable and chaotic dynamics in quantum systems. It rose to prominence in the quantum chaos literature in a famous paper of Petr \u{S}eba \cite{Seba} which dealt with the closely related case of rectangular billiards. The model first appeared in solid state physics \cite{KronigPenney} in the 1930s to explain electronic band structure and conductivity in solid crystals. Many applications arose in nuclear physics throughout the 1960s and 1970s, see for instance \cite{FaddeevBerezin}. The purpose of this article is to give an introduction to this important model which belongs to the class of pseudo-integrable systems and to report on some recent progress in this field. The reader will also be introduced to some important open problems.

\subsection{Kronig-Penney model}
In 1931 Kronig and Penney \cite{KronigPenney} studied the quantum mechanics of an electron in a periodic crystal lattice with the goal of understanding the conductivity properties of solid crystals.

They introduced the periodic 1D Hamiltonian
\begin{equation}
\hat{H}_{KP}=-\frac{d^2}{dx^2}+V_0\sum_{k\in\Z}\chi_{[-\tfrac{a}{2},\tfrac{a}{2}]}(x-k), \quad V_0>0, \quad 0<a\ll 1
\end{equation}
where $\chi$ denotes the characteristic function. According to Bloch theory, we have the decomposition $$L^2(\R)=\int_\oplus d\theta \scrH_\theta, \quad \theta\in[0,2\pi)$$ where $\scrH_\theta$ is the space of quasiperiodic functions with quasimomentum $e^{\i\theta}$:
$$\scrH_\theta=\{\psi\in C^\infty(\R) \mid \psi(x)=e^{\i k\theta }\psi(x+k), \quad k\in\Z\}$$
Let us consider the special case of periodic boundary conditions $\theta=0$.
To simplify the Hamiltonian ${\hat{H}_{KP}}|_{\scrH_0}$ it is convenient to take the limit $$a\searrow0, \quad \alpha=V_0 a=const.$$

Let $f\in C^\infty(\R/\Z)$. The calculation $$\int_{\R/\Z}V_0\chi_{[-\tfrac{a}{2},\tfrac{a}{2}]}(x)f(x)dx=\frac{\alpha}{a}\int_{-a/2}^{a/2}f(x)dx \to \alpha f(0), \quad a\searrow 0$$ shows that the Hamiltonian ${\hat{H}_{KP}}|_{\scrH_0}$ converges in the distributional sense to a singular rank-one perturbation of the 1D Laplacian 
\begin{equation}
-\frac{d^2}{dx^2}+V_0\chi_{[-\tfrac{a}{2},\tfrac{a}{2}]} \to H_\alpha=-\frac{d^2}{dx^2}+\alpha\left\langle \delta_0,\cdot \right\rangle\delta_0, \quad a\searrow 0.
\end{equation} 
The operator $H_\alpha$ can be realised rigorously by using von Neumann's self-adjoint extension theory. We will be interested in studying the analogues of the operator $H_\alpha$ on 2D and 3D tori.

\subsection{\u{S}eba billiard}
Let $R$ be a rectangle with side lengths $a$, $b$. We define the aspect ratio of $R$ as the quotient $(a/b)^2$. In a 1990 paper \cite{Seba} Petr \u{S}eba studied the operator 
\begin{equation}
H_\alpha=-\Delta+\alpha\left\langle\delta_{x_0},\cdot\right\rangle\delta_{x_0}, \quad x_0\in R, \; \alpha\in\R\setminus\{0\}
\end{equation}
on a rectangle with irrational aspect ratio and Dirichlet boundary conditions. \u{S}eba's motivation was to find a quantum system which displayed the features of quantised chaotic systems such as quantum ergodicity and level repulsion, yet whose classical dynamics was close to integrable. 

As was pointed out later by Shigehara \cite{Shigehara1} the energy levels obtained in Seba's quantisation do not repell each other, in fact careful numerical experiments conducted by Shigehara show that the spacing distribution coincides with that of the Laplacian which is conjectured to be Poissonian. We will discuss rigorous mathematical results in this direction in section 5. Shigehara suggested a different quantisation in his paper which should produce energy levels which display level repulsion. In the present paper we refer to Seba's quantisation as ``weak coupling'' and to Shigehara's as ``strong coupling''. A detailled discussion of these two different quantisations is given in section 3. 

In the present paper we will deal with a system closely related to the Seba billiard -- a point scatterer on a flat torus (which means periodic boundary conditions), however, the results which will be presented can probably be easily extended to rectangular domains with Dirichlet or Neumann boundary conditions.

{\bf Acknowledgements:}
I would like to thank Zeev Rudnick and Stephane Nonnenmacher for many helpful comments and suggestions that have led to the improvement of this paper.

\section{Quantisation of a point scatterer}

\subsection{Self-adjoint extension theory}
We consider a rectangle with side lengths $2\pi a$, $2\pi/a$, where $a>0$, and identify opposite sides to obtain the torus $\T^2=\R^2/2\pi\scrL_0$ where $\scrL_0=\Z(a,0)\oplus\Z(0,1/a)$.
We want to study the formal operator 
\begin{equation}
H_\alpha=-\Delta+\alpha\delta_{x_0}, \quad \alpha\in\R\setminus\{ 0\}, \quad x_0\in\T^2.
\end{equation} 

To treat $H_\alpha$ rigorously we will employ von Neumann's theory of self-adjoint extensions. For an introduction to this standard machinery see \cite{Z}. The main idea is to restrict $H_\alpha$ to a domain where we understand how it acts -- functions which vanish at the position of the scatterer and therefore do not ``feel" its presence. 

We denote by 
\begin{equation}
D_0=C^\infty_0(\T^2\setminus\{x_0\})
\end{equation}
the domain of $C^\infty$-functions which vanish in a neighbourhood of $x_0$. Clearly $H_\alpha|_{D_0}=-\Delta|_{D_0}$. We denote $-\Delta_0=-\Delta|_{D_0}$. The restricted Laplacian $-\Delta_0$ is a symmetric operator, however it is not self-adjoint. By restricting $-\Delta$ to the domain $D_0$ we are enlarging the domain of its adjoint.
Therefore we have $D_0=Dom(-\Delta_0)\subsetneq Dom(-\Delta_0^*)$. A simple computation of the adjoint $-\Delta_0^*$ shows that its domain is given by
\begin{equation}
Dom(-\Delta_0^*)=\{f\in L^2(\T^2) \mid \exists C\in\C: \; \Delta f + C\delta_{x_0}\in L^2(\T^2)\}.
\end{equation}
We have the following definition.
\begin{defn}
The deficiency spaces of a symmetric densely defined operator $A$ are given by the kernels
\begin{equation}
\scrK_{\pm}=\ker\{A^*\pm\i\}.
\end{equation}
The deficiency indices of $A$ are defined as $n_+=\dim\scrK_+$ and $n_-=\dim\scrK_-$. If $n_+=n_-=0$, then we say that $A$ is essentially self-adjoint.
\end{defn}
For $\lambda\notin\sigma(-\Delta)$ denote by $G_\lambda(x,y)$ the corresponding Green's function, namely the integral kernel of the resolvent
\begin{equation}
\frac{1}{\Delta_x+\lambda}f(x)=\int_{\T^2}G_\lambda(x,y)f(y)dy
\end{equation}
and therefore we have the following distributional identity
\begin{equation}
G_\lambda(x,y)=\frac{1}{\Delta_x+\lambda}\delta_y(x).
\end{equation}

Indeed, if we compute the deficiency elements of $-\Delta_0^*$ we have to solve
\begin{equation}
0=(-\Delta_0^*\pm\i)f=(-\Delta\pm\i)f+C\delta_{x_0}
\end{equation}
for some $C\in\C$. This shows that the deficiency spaces are spanned by the Green's functions $G_{\pm\i}(x,x_0)$. We thus have
\begin{equation}
Dom(-\Delta_0^*)=\overline{D_0}\oplus_{\perp}\scrL\{G_{\i}\}\oplus_{\perp}\scrL\{G_{-\i}\}
\end{equation}
where the orthogonal decomposition is with respect to the graph inner product $\left\langle f,h \right\rangle_g=\left\langle f,h \right\rangle+\left\langle \Delta_0^*f,\Delta_0^*h \right\rangle$ and the closure is taken with respect to the associated graph norm $\|f\|_g=\left\langle f,f \right\rangle_g$. 

The following theorem is due to von Neumann.
\begin{thm}
Let $A$ be a densely defined symmetric operator. If $A$ has deficiency indices $n_+=n_-=n\geq1$, then there exists a family of self-adjoint extensions which is parametrised by $U(n)$, the group of unitary maps on $\C^n$. The domain of the extension $A_U$ is given by
\begin{equation}
\begin{split}
Dom(A_U)=\{f\in Dom(A^*) \mid f=g+\left\langle G_\i, v \right\rangle+&\left\langle G_{-\i}, Uv \right\rangle, \\
 &g\in\overline{Dom(A)}, v\in\C^n\}
 \end{split}
\end{equation}
where $G_\i$, $G_{-\i}$ are the vectors whose entries are the deficiency elements and the closure of $Dom(A)$ is taken with respect to the graph norm of $A$. The operator $A_U=A^*|_{Dom(A_U)}$ is essentially self-adjoint.
\end{thm}

In our case we have for $\varphi\in(-\pi,\pi]$
\begin{equation}
\begin{split}
D_\varphi=Dom(-\Delta_\varphi)=&\{f\in Dom(-\Delta_0^*) \mid \\
 & \quad f=g+cG_\i+c e^{\i\varphi}G_{-\i},\;g\in\overline{D_0},c\in\C\}
 \end{split}
\end{equation}
and $-\Delta_\varphi$ is the restriction of $-\Delta_0^*$ to this domain. Functions $f\in Dom(-\Delta_\varphi)$ satisfy
$$\exists C\in\C: \quad \Delta f+C\delta_{x_0}\in L^2(\T^2)$$ and near $x_0$ they have the asymptotic $$f(x)=C\left(\frac{\log|x-x_0|}{2\pi}+\tan\frac{\varphi}{2}\right)+o(1).$$
The extension associated with the choice $\varphi=\pi$ is just the self-adjoint Laplacian on $C^\infty(\T^2)$. We will be interested in studying the extensions for $\varphi\in(-\pi,\pi)$.

\subsection{The quantisation condition}
An orthonormal basis of eigenfunctions of the Laplacian on $\T^2=\R^2/2\pi\scrL_0$ is given by the complex exponentials $$\psi_\xi(x)=\frac{1}{2\pi}e^{\i\left\langle x,\xi \right\rangle} \quad \text{where}\; \xi\in\scrL$$ and $\scrL$ is the dual lattice of $\scrL_0$:
\begin{equation}
\scrL=\{\xi\in\R \mid \forall \eta\in\scrL_0: \; \left\langle \xi,\eta \right\rangle\in\Z\}=\{(m/a,na) \mid m,n\in\Z\}
\end{equation}
The eigenvalue of $\psi_\xi$ is given by $|\xi|^2$. We introduce the set of distinct norms
$$\scrN=\{0=n_0<n_1<n_2<\cdots\}$$ and denote the multiplicity of $n\in\scrN$ by
\begin{equation}
r_\scrL(n)=\#\{\xi\in\scrL \mid n=|\xi|^2\}.
\end{equation}
If $\scrL$ is a rational lattice, i. e. if $a^4\in\Q$, then the multiplicities can be large and we have the bound (see e.g. \cite{OraveczRudnickWigman}, Lemma 7.2)
\begin{equation}
r_\scrL(n)\ll_\epsilon n^\epsilon.
\end{equation}

\begin{figure}
\includegraphics{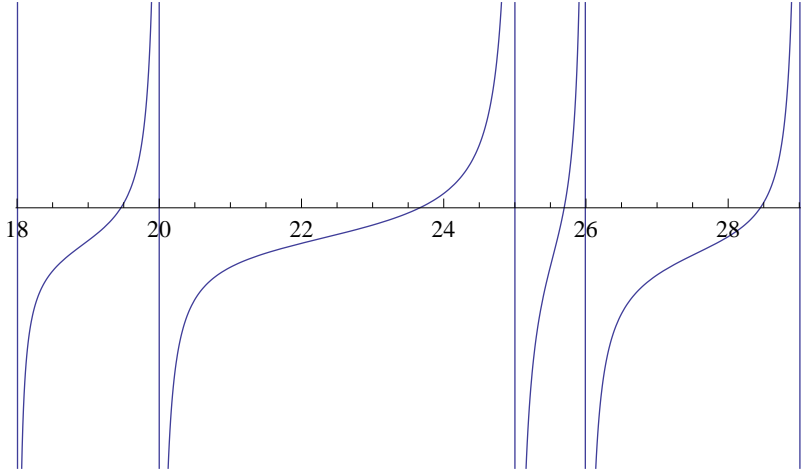}
\caption{The function on the RHS of eq. \eqref{quant cond} for the square lattice.}
\end{figure}

We have the following lemma which can be found in the standard literature on point scatterers, or in the appendix to the paper \cite{RU}.
\begin{lem}\label{quantisation lemma}
Let $\varphi\in(-\pi,\pi)$. We have that $\lambda\notin\sigma(-\Delta)$ is an eigenvalue of $-\Delta_\varphi$ on $\T^2$ iff
\begin{equation}\label{quant cond}
\sum_{j=0}^\infty r_\scrL(n_j)\left(\frac{1}{n_j-\lambda}-\frac{n_j}{n_j^2+1}\right)=c_0\tan(\varphi/2)
\end{equation}
where $$c_0=\sum_{j=0}^\infty \frac{r_\scrL(n_j)}{n_j^2+1}$$
and the corresponding eigenfunction is a multiple of the Green's function $G_\lambda(x,x_0)$ for which we have the $L^2$-identity
\begin{equation}\label{spec expansion}
G_\lambda(x,x_0)=\frac{1}{4\pi}\sum_{j=0}^\infty\frac{1}{n_j-\lambda}\sum_{|\xi|^2=n_j} e^{\i\left\langle \xi,x-x_0 \right\rangle}.
\end{equation}
\end{lem}
The eigenfunctions in the lemma are not the only eigenfunctions of $-\Delta_\varphi$. If $r_\scrL(n_j)>1$ (which in our case happens for all $n_j>0$), then $n_j$ appears in the spectrum of $-\Delta_\varphi$ with multiplicity $r_\scrL(n_j)-1$. The associated eigenfunctions are superpositions of Laplacian eigenfunctions and vanish at $x_0$ and therefore do not feel the effect of the scatterer. 
In the present article we will only be interested in the new eigenvalues, which are solutions to \eqref{quant cond}, and the associated new eigenfunctions.

\section{Weak coupling vs. strong coupling}
Two different quantisations appear in the literature on point scatterers, and we will follow the terminology used by Shigehara et al. in the papers \cite{Shigehara1,Shigehara2,ShigeharaCheon3D,Shigehara3} in referring to the two models as {\em weak} and {\em strong} coupling. The purpose of this section is to explain how these different quantisations arose in the literature and give a derivation of the strong coupling quantisation on a 2D torus.

In his famous paper on wave chaos in a singular billiard \cite{Seba} \u{S}eba considered a point scatterer on a rectangle with irrational aspect ratio and Dirichlet boundary conditions. \u{S}eba computed the spectrum of $-\Delta_\varphi$ by solving the equation \eqref{quant cond} numerically and a plot of the level spacing distribution seemed to suggest level repulsion. 

In 1994 Shigehara came to investigate this question and following careful numerical investigations he observed \cite{Shigehara1} that the eigenvalue spacing distribution of the self-adjoint extension $-\Delta_\varphi$ seemed to coincide with that of the Laplacian, which according to the conjecture of Berry and Tabor \cite{BerryTabor} is believed to be Poissonian. Shigehara observed that the apparent ``weakness" of the point scatterer in \u{S}eba's quantisation could be corrected by adjusting the parameter $\varphi$ in a suitable way as the eigenvalue $\lambda$ tends to infinity. 


The operator can be realised by employing the self-adjoint extension theory discussed in the previous section. The formal operators 
\begin{equation}
-\Delta+\alpha\delta_{x_0}, \quad \alpha\in\R\setminus\{0\}
\end{equation} 
are associated with the family of extensions 
\begin{equation}
-\Delta_\varphi, \quad \varphi\in(-\pi,\pi).
\end{equation}
It is well known that in 1D there is an exact relation which links the physical coupling constant $\alpha\neq 0$ and the parameter $\varphi\in(-\pi,\pi)$. The situation is more complicated in 2 or 3D. A relation can be derived (at a physical level of rigour) from the scattering problem for a spherical scatterer (cf. for instance \cite{SebaExner}) by shrinking its diameter to zero. 

One obtains a relation which in contrast to the 1D case contains a logarithmic divergence in the spectral parameter $\lambda$
\begin{equation}
\frac{1}{\alpha}=C_1\tan\frac{\varphi}{2}+C_2\log\lambda
\end{equation}
for certain real constants $C_1$, $C_2$ which depend on the domain. 

So a fixed choice of $\varphi$ corresponds to weak coupling, 
\begin{equation}
\alpha\sim\frac{1}{C_2\log\lambda}.
\end{equation}
On the other hand a fixed physical coupling constant $\alpha=const.\neq 0$ would require a renormalisation of the parameter $\varphi$, which means $\varphi=\varphi_\lambda$ should be allowed to depend on $\lambda$ as $\lambda\to\infty$, 
\begin{equation}\label{renorm}
C_1\tan\frac{\varphi_\lambda}{2}\sim-C_2\log\lambda.
\end{equation}


The renormalisation condition \eqref{renorm} is equivalent to a different quantisation condition for a point scatterer (cf. for instance \cite{BogomolnyGerlandSchmit}) which only takes into account the physically relevant energies in the summation,
\begin{equation}\label{renorm quant cond}
\sum_{|n_j-\lambda|<\lambda^\delta}r_\scrL(n_j)\left\{\frac{1}{n_j-\lambda}-\frac{n_j}{n_j^2+1}\right\}=\frac{1}{\alpha}
\end{equation}
where $\alpha\neq0$ is the physical coupling constant.
We require the following lemma which gives an asymptotic for the sum over the energies outside the interval $[\lambda-\lambda^\delta,\lambda+\lambda^\delta]$.
\begin{lem}\label{truncation}
Consider a general torus $\T^2$. Let $\delta\in(\theta,1)$, where $\theta=\frac{131}{416}$. We have the asymptotic
\begin{equation}
\sum_{|n_j-\lambda|\geq\lambda^\delta}r_\scrL(n_j)\left\{\frac{1}{n_j-\lambda}-\frac{n_j}{n_j^2+1}\right\}=-\pi\log\lambda+O(1)
\end{equation}
\end{lem}
\begin{proof}
We have
\begin{equation}
\begin{split}
&\sum_{|n_j-\lambda|\geq\lambda^\delta}r_\scrL(n_j)\left\{\frac{1}{n_j-\lambda}-\frac{n_j}{n_j^2+1}\right\}\\
&=\sum_{|n_j-\lambda|\geq\lambda^\delta}r_\scrL(n_j)\left\{\frac{1}{n_j-\lambda}-\frac{1}{n_j}\right\}+O(1).
\end{split}
\end{equation}
We will use the circle law
\begin{equation}
\sum_{n_j\leq x}r_\scrL(n_j)=\pi x+O_\epsilon(x^{\theta+\epsilon})
\end{equation}
where the best known exponent $\theta=\tfrac{131}{416}$ is due to Huxley \cite{Huxley} and the optimal exponent is expected to be $\theta=\tfrac{1}{4}$ (Gauss's circle problem).

Summation by parts allows us to compare a lattice sum with an integral
\begin{equation}
\sum_{\alpha\leq|\xi|^2\leq\beta}f(|\xi|^2)=\pi\int_\alpha^\beta f(t)dt+O(\alpha^\theta f(\alpha)+\beta^\theta f(\beta))
+O\left(\int_\alpha^\beta |f'(t)|t^\theta dt\right)
\end{equation}
where $f$ is some differentiable function.

We obtain for the first sum
\begin{equation}
\begin{split}
\sum_{\lambda+\lambda^\delta\leq n_j}r_\scrL(n_j)\left\{\frac{1}{n_j-\lambda}-\frac{1}{n_j}\right\}
&=\pi\int_{\lambda+\lambda^\delta}^\infty\left\{\frac{1}{x-\lambda}-\frac{1}{x}\right\}dx+O(\lambda^{\theta-\delta})\\
&=(1-\delta)\pi\log\lambda+O(\lambda^{\theta-\delta}).
\end{split}
\end{equation}
Similarly for the second sum
\begin{equation}
\begin{split}
\sum_{n_j\leq\lambda-\lambda^\delta}r_\scrL(n_j)\left\{\frac{1}{n_j-\lambda}-\frac{1}{n_j}\right\}
&=\pi\int^{\lambda-\lambda^\delta}_1\left\{\frac{1}{x-\lambda}-\frac{1}{x}\right\}dx+O(1)\\
&=(\delta-2)\pi\log\lambda +O(1).
\end{split}
\end{equation}
\end{proof}

With the help of the Lemma above we can now show that the quantisation \eqref{renorm quant cond} is in fact equivalent to the strong coupling quantisation given by the renormalisation condition \eqref{renorm}.

We obtain
\begin{equation}
\sum_{j=0}^\infty r_\scrL(n_j)\left\{\frac{1}{n_j-\lambda}-\frac{n_j}{n_j^2+1}\right\}=-\pi\log\lambda+\frac{1}{\alpha}+O(1)
\end{equation}
which forces the renormalisation
\begin{equation}
c_0\tan\frac{\varphi_\lambda}{2}\sim-\pi\log\lambda.
\end{equation}

\section{Statistics of wave functions}

Let $\{\lambda_j^\varphi\}$ denote the solutions to the spectral equation \eqref{quant cond} (where $\varphi$ is fixed). The eigenvalues $\{\lambda_j^\varphi\}$ interlace with the norms $\{n_j\}$ in the following way
\begin{equation}
\lambda_0^\varphi<n_0<\lambda_1^\varphi<n_1<\cdots<\lambda_j^\varphi<n_j<\cdots 
\end{equation}
We denote the solutions to the strong coupling quantisation condition \eqref{renorm quant cond} by $\{\lambda_j^{str.}\}$. Similarly as above the eigenvalues $\{\lambda_j^{str.}\}$ interlace with the norms $\{n_j\}$. The corresponding eigenfunctions are the Green's functions $\{G_{\lambda_j^{str.}}\}$.
We are interested in the statistical behaviour of the eigenfunctions $\{G_{\lambda_j^\varphi}\}$ and $\{G_{\lambda_j^{str.}}\}$ in the limit as the eigenvalue tends to infinity.

Let $a\in C^\infty(S^*\T^2)$ where $S^*\T^2 \simeq \T^2 \times S^1$. Let $(x,\phi)\in\T^2\times S^1$ and expand $a\in C^\infty(\T^2\times S^1)$ into a Fourier series
\begin{equation}
a(x,\phi)=\sum_{\zeta\in\scrL, k\in\Z}\hat{a}(\zeta,k)e^{\i\left\langle \zeta,x \right\rangle+\i k\phi}
\end{equation}

In order to quantise the classical symbol $a$ it is convenient to realise the torus as the quotient $\T^2=\C/2\pi\scrL_0[\i]$, where $$\scrL_0[i]=\{x_1+\i x_2 \mid (x_1,x_2)\in\scrL_0\}.$$ For $\zeta\in\scrL_0$ we define by $\hat{\zeta}$ the corresponding element of $\scrL_0[\i].$
We associate with the symbol $a$ the zeroth order pseudo-differential operator $Op(a):L^2(\T^2)\to L^2(\T^2)$ defined on the Fourier transform side by
\begin{equation}
\widehat{(Op(a)f)}(\xi)=\sum_{\zeta\in\scrL, k\in\Z}\hat{a}(\zeta,k)\left(\frac{\hat{\xi}}{|\xi|}\right)^k\hat{f}(\xi-\zeta), \quad \xi\neq0
\end{equation}
and 
\begin{equation}
\widehat{(Op(a)f)}(0)=\sum_{\zeta\in\scrL}\hat{a}(\zeta,0)\hat{f}(\xi-\zeta)
\end{equation}
where
\begin{equation}
f(x)=\sum_{\xi\in\scrL}\hat{f}(\xi)e^{\i\left\langle \xi,x \right\rangle}.
\end{equation}

Let $\lambda\in\R\setminus\sigma(-\Delta)$. Denote the $L^2$-normalised Green's function by $g_\lambda=G_\lambda/\|G_\lambda\|_2$. We would like to study the behaviour of the matrix elements $$\left\langle Op(a) g_{\lambda_j^\varphi}, g_{\lambda_j^\varphi}\right\rangle$$ in the limit as $j\to\infty$. Similarly we would like to study the matrix elements of $Op(a)$ in the strong coupling quantisation, $$\left\langle Op(a) g_{\lambda_j^{str.}}, g_{\lambda_j^{str.}}\right\rangle.$$

We first study the behaviour of the eigenfunction in position space. If we take a classical symbol on position space $a\in C^\infty(\T^2)$, then the operator $Op(a)$ is simply given by multiplication
\begin{equation}
\left\langle Op(a)g_\lambda,g_\lambda \right\rangle=\int_{\T^2} a|g_\lambda|^2 d\mu.
\end{equation}

We have the following theorem, which is proven in \cite{RU}. In the paper the theorem is stated for the specific case of the eigenfunctions $g_{\lambda_j^\varphi}$ of the weakly coupled point scatterer. However, the result holds for any increasing sequence of numbers which interlaces with the norms $\{n_j\}$. In particular it also applies to the eigenfunctions $g_{\lambda_j^{str.}}$ of a strongly coupled point scatterer. We state the theorem in full generality.
\begin{thm} \label{wave function stats}{\bf (Rudnick-U., 2012)}\\
Let $\T^2=\R^2/2\pi\scrL_0$ be a general flat torus. Let $a\in C^\infty(\T^2)$. Recall that $g_\lambda=G_\lambda/\|G_\lambda\|_2$ denotes the $L^2$-normalised Green's function.
For any increasing sequence of numbers $\Lambda=\{\lambda_j\}$ which interlaces with the norms $\scrN=\{n_j\}$ there exists a density one subsequence $\Lambda_\infty\subset\Lambda$ such that
\begin{equation}
\int_{\T^2} a|g_{\lambda_j}|^2 d\mu \to \frac{1}{\area(\T^2)}\int_{\T^2} a\;d\mu
\end{equation}
as $j\to\infty$ along $\Lambda_\infty$.
\end{thm}
\begin{remark}
In particular this theorem implies that {\bf both} in the weak {\bf and} strong coupling regimes the eigenfunctions equidistribute in position space. Also note that the result holds for both rational and irrational lattices. The analogous result was proved in \cite{Yesha} for the standard 3D torus and general 3D tori which satisfy certain irrationality conditions.
\end{remark}
The proof of Theorem \ref{wave function stats} uses the explicit formula \eqref{spec expansion} for the Green's function and by approximating the Green's function by a sum over a polynomial size interval the problem can be translated into a number theoretical problem about the well-spacedness of lattice points in thin annuli.

In a recent paper \cite{MarklofR} Marklof and Rudnick have shown that for rational polygons a full density of eigenfunctions of the Laplacian equidistributes in position space. In particular their result applies to the square torus. Their proof uses Egorov's theorem.

Quantum ergodicity does of course not hold for the Laplacian eigenfunctions on the torus. If one chooses an eigenbasis of plane waves, the eigenfunctions are obviously localised in momentum. It is an interesting question to ask if the presence of a (weakly or strongly coupled) point scatterer can change this. Curiously, the answer depends on the arithmetic properties of the lattice $\scrL_0$. 

For both a weakly and strongly coupled point scatterer it is possible \cite{KU} to prove quantum ergodicity for $\scrL_0=\Z^2$ and it is likely that one can generalise this to any rational lattice. On the other hand one can disprove quantum ergodicity in the irrational case. The failure of quantum ergodicity in the closely related case of the Seba billiard (which has an irrational aspect ratio) was already conjectured by Berkolaiko, Keating and Winn \cite{BerkolaikoKeatingWinn}. 

In \cite{KeatingMarklofWinn} Keating, Marklof and Winn prove under assumptions on the spectrum (which are consistent with the Berry-Tabor conjecture) that there exists a positive density subsequence of eigenfunctions for the Seba billiard (weak coupling) which become localised around two Laplacian eigenfunctions. In particular it follows that this subsequence becomes localised in the high energy limit, therefore disproving Quantum Ergodicity. This phenomenon is in some sense similar to scarring on unstable periodic orbits of the classical system. The authors construct a sequence of quasimodes which converge to the real eigenfunctions in order to obtain their results. It would be interesting to prove this result without any assumptions on the spectrum. In the closely related case of a weakly coupled point scatterer on an irrational 2D torus the failure of Quantum Ergodicity can be proved unconditionally \cite{KU}. Another interesting question is to see if one can classify quantum limits which are localised in momentum without any assumptions on the spectrum.

{\bf Open problems:} {\em  It would be interesting to try to prove the analogous statement for the strong coupling limit: Can one disprove quantum ergodicity for irrational lattices?}

The microcolal lift $d\mu_j$ of the measures $|g_{\lambda_j^\varphi}|^2 d\mu$ is defined by the identity 
\begin{equation}
\left\langle Op(a)g_{\lambda_j^\varphi},g_{\lambda_j^\varphi}\right\rangle=\int_{S^*\T^2}a\, d\mu_j, \quad a\in C^\infty(S^*\T^2).
\end{equation}
The quantum limits are the limit points of the sequence $d\mu_j$ in the weak-* topology. In \cite{Jakobson} Jakobson classified the quantum limits for the Laplacian on the square torus. One can pose the same problem for a point scatterer.

{\bf Open problems:} {\em
What are the quantum limits for a point scatterer on the square torus in the weak and strong coupling regimes? How about general rational and irrational 2D tori? How about 3D tori?}

\section{Spectral statistics}

One of the main observations in \u{S}eba's paper \cite{Seba} was level repulsion for a point scatterer in a rectangular billiard with Dirichlet boundary conditions and irrational aspect ratio. \u{S}eba computed the eigenvalues numerically and plotted the spacing distribution which did indeed reveal some form of level repulsion and he initially conjectured that the spacing distribution should coincide with that of the Gaussian Orthogonal Ensemble (GOE) in Random Matrix Theory. 

As Shigehara \cite{Shigehara1} discovered later, level repulsion, as observed by \u{S}eba, could only be expected in the strong coupling regime, where the extension parameter is normalised in a suitable way as the eigenvalue tends to infinity. It is likely that Seba carried out a truncation of the spectral equation (similar to the one in Lemma \ref{truncation}) when performing the numerics, so effectively he calculated the eigenvalues for the strong coupling quantisation.

Based on careful numerics, Shigehara predicted that one should recover Poissonian level statistics in the weak coupling regime, as was expected for the unperturbed Laplacian on an irrational rectangle. In other words, the effect of the scatterer in this regime was too weak to have an impact in the high energy limit. In 3D, however, the situation is very different. Numerics suggest that the eigenvalues of a fixed self-adjoint extension obey intermediate level statistics, just as in the strong coupling regime in 2D \cite{BogomolnyGerlandSchmit}.

In a 1999 paper \cite{BogomolnyGerlandSchmit} Bogomolny, Gerland and Schmit argued that the spacing distribution was close to a semi-Poissonian distribution. The \u{S}eba billiard thus belongs to an intermediate class of systems which shows a weaker form of level repulsion than chaotic systems. Another example of such intermediate systems are flat surfaces with conic singularities. Intermediate statistics have also been observed for certain families of quantum maps \cite{GiraudMarklofOKeefe} and near the transition point in the Anderson model in 3D (cf. \cite{Guhr}, section 6. 1. 3., p. 332-3).

There is also a recent paper by St\"ockmann, Kuhl and Tudorovskiy \cite{Stoeckmann} which claims to revise the results of Bogomolny, Gerland and Schmit in \cite{BogomolnyGerlandSchmit}. The paper contains heuristic mathematical arguments and numerics which show a transition from level repulsion to Poissonian spacing statistics for the Seba biliard in the high energy limit. However, the quantisation condition used in \cite{Stoeckmann} can be shown to correspond to the weak coupling regime -- as opposed to the strong coupling regime considered by Bogomolny, Gerland and Schmit. The Poissonian statistics observed in \cite{Stoeckmann} are therefore hardly surprising, as they were already predicted by Shigehara in \cite{Shigehara1}. A rigorous mathematical proof that the spacing distribution of a weakly coupled point scatterer coincides with that of the Laplacian on the 2D torus is given in the paper \cite{RU2} and we will discuss this result and others below.

There are many open questions regarding the spectral statistics of point scatterers on flat tori. Even in the case of the unperturbed Laplacian, little can be said rigorously. Let $$\delta_j=n_j-n_{j-1}>0$$ and define $$N_{\scrL_0}(x)=\#\{ n_j\leq x\}.$$ If $\scrL_0=\Z^2$, we have the following asymptotic which is due to Landau \cite{Landau}
\begin{equation}
N_{\Z^2}(x)\sim \frac{Bx}{\sqrt{\log x}}, \quad B=0.764...
\end{equation}
If $\scrL_0$ is an irrational lattice, then the multiplicity of the Laplacian eigenvalues is on average $4$ (for a generic vector $\xi\in\scrL_0[\i]$ the other three choices are $-\xi$, $\bar{\xi}$ and $-\bar{\xi}$). Indeed we have in this case
\begin{equation}
N_{\scrL_0}(x)\sim \frac{\pi}{4}x.
\end{equation}

We define the mean spacing up to a threshold $x$ by
\begin{equation}
\ave{\delta_j}_x=\frac{1}{N_{\scrL_0}(x)}\sum_{n_j\leq x}\delta_j
\end{equation}
and observe that
\begin{equation}
\ave{\delta_j}_x \sim \frac{x}{N_{\scrL_0}(x)} \sim
\begin{cases}
\frac{4}{\pi}, \quad\text{if $\scrL_0$ is irrational}\\
\\
B^{-1}\sqrt{\log x}, \quad\text{if $\scrL_0=\Z^2$}
\end{cases}
\end{equation}
We introduce the mean normalised spacings
\begin{equation}
\hat{\delta}_j=\frac{\delta_j}{\ave{\delta_j}}, \quad n_j\leq x.
\end{equation}
Analogously let $\delta_j^\varphi=\lambda_j^\varphi-\lambda_{j-1}^\varphi>0$ and we introduce the mean normalised spacings
\begin{equation}
\hat{\delta}^\varphi_j=\frac{\delta^\varphi_j}{\ave{\delta_j}}, \quad n_j\leq x
\end{equation}
and it is easy to see that $\ave{\delta_j^\varphi}_x \sim \ave{\delta_j}_x$ as $x\to\infty$ (recall that $\delta_j^\varphi$ are the spacings of the new eigenvalues which interlace with the norms $n_j$).

We have the following conjecture which is due to Berry and Tabor \cite{BerryTabor} and is expected to hold for a generic classically integrable quantum system. We will only state it in the special case of flat 2D tori.
\begin{conj} 
Let $\scrL_0$ be a Diophantine irrational lattice. Then the mean normalised spacings $\{\hat{\delta_j}\}_{n_j\leq x}$ have a Poissonian distribution of mean 1 in the limit as $x\to\infty$. So for any $h\in C^\infty(\R_+)$
\begin{equation}
\lim_{x\to\infty}\frac{1}{N_{\scrL_0}(x)}\sum_{n_j\leq x}h(\hat{\delta}_j)=\int_0^\infty h(s)e^{-s}ds
\end{equation}
\end{conj}
\begin{remark}
We even expect this conjecture to hold in the case of the square torus (Important: we are ignoring multiplicities).
\end{remark}
This conjecture remains out of reach, even a proof of the existence of the spacing distribution is not available. However, some progress has been made with regard to the pair correlation function \cite{EMM,Marklof}, in cases where the aspect ratio satisfies certain diophantine properties. 

It is an interesting problem to study the effect of the scatterer on the distribution of the spacings $\hat{\delta}_j^\varphi$. We introduce the difference between Laplacian eigenvalues and the neighbouring eigenvalue of the point scatterer: $$d_j=n_j-\lambda_j^\varphi>0$$ and we denote its mean by 
\begin{equation}
\ave{d_j}_x=\frac{1}{N_{\scrL_0}(x)}\sum_{n_j\leq x}d_j.
\end{equation}

One has the following bound on the mean of $d_j$ which shows that on average the eigenvalues $n_j$ and $\lambda_j^\varphi$ ``clump" together. The result is derived in \cite{RU2} and the main tool is an exact trace formula for a point scatterer, similar to the trace formula proved in \cite{U}.
\begin{thm} {\bf (Rudnick-U., 2012)}\label{clumpingthm}\\
Let $\scrL_0$ be a general lattice. Fix $\varphi\in(-\pi,\pi)$. We have the bound
\begin{equation}\label{clumping}
\ave{d_j}_x \ll \frac{\ave{\delta_j}_x}{\log x}
\end{equation}
\end{thm}
\begin{remark}
This implies that if either the limiting distribution of $\hat{\delta}_j$ or $\hat{\delta}_j^\varphi$ exist, then they must coincide. By use of the trace formula proved in \cite{U} it is possible to extend this result to surfaces of constant negative curvature.
\end{remark}
The situation is very different in 3D. Denote $\scrL_0^3=\Z(a,0,0)\oplus\Z(0,b,0)\oplus\Z(0,0,1/ab)$, $a,b>0$ and consider the torus $\T^3=\R^3/2\pi\scrL_0^3$. An orthonormal basis of Laplacian eigenfunctions is given by the exponentials $$\psi_\xi(x)=\frac{1}{\sqrt{8}\pi^{3/2}}e^{\i\left\langle \xi,x \right\rangle}, \quad \xi\in\scrL^3$$ where $\scrL^3$ is the dual of $\scrL_0^3$, and the eigenvalue is given by the norm $|\xi|^2$. We denote the set of distinct norms by $\scrN_3$. 

Let $x_0\in\T^3$. In 3D the operator $-\Delta_0=-\Delta|_{D_0}$, where $D_0=C^\infty_0(\T^3\setminus\{x_0\})$, has deficiency indices $(1,1)$ \footnote{In dimension $d\geq 4$ the analogous operator $-\Delta_0$ is essentially self-adjoint.} and we denote the self-adjoint extensions by $-\Delta_\varphi$. The new eigenvalues of $-\Delta_\varphi$ (which interlace with the norms $\scrN_3$) are denoted by $\eta_j^\varphi$ and we define the differences $\delta_j=n_j-n_{j-1}$ and $d_j=n_j-\eta_j^\varphi$. As opposed to Theorem \ref{clumpingthm} the eigenvalues of the point scatterer on a 3D torus seem to lie on average in the middle of the neigbouring Laplacian eigenvalues. We have the following result which is consistent with intermediate statistics. The proof is given in \cite{RU2} and the methods are the same as in the case of Theorem \ref{clumpingthm}.
\begin{thm} {\bf (Rudnick-U., 2012)}\\\label{intermediate}
Fix $\varphi\in(-\pi,\pi)$. For a general lattice $\scrL_0^3$ we have the asymptotic
\begin{equation}
\ave{d_j}_x\sim \frac{1}{2}\ave{\delta_j}_x.
\end{equation}
\end{thm}
It would be desirable to obtain information about the variance of $d_j$. Theorems \ref{clumpingthm} and \ref{intermediate} are proved using a trace formula for the point scatterer of the type proved in \cite{U}. However, this approach fails for higher moments.\\

{\bf Open problems:} Consider a 3D torus.
Can one obtain an asymptotic for $\ave{d_j^2}_x$ or a useful bound such as $\ave{d_j^2}_x<c_0\ave{\delta_j}_x^2$ for some $c_0<\tfrac{1}{2}$? What about higher moments?

\end{document}